\documentclass[3p, authoryear, twocolumn]{elsarticle} 

\usepackage{graphicx}
\usepackage{hyperref}

\usepackage[utf8]{inputenc}
\usepackage[ruled,linesnumbered]{algorithm2e}

\usepackage[utf8]{inputenc}
\usepackage[T1]{fontenc}
\usepackage[english]{babel}
\usepackage{xspace}
\usepackage{amsmath,amssymb,amsthm}
\usepackage{stmaryrd}
\usepackage{lipsum}

\usepackage{pgf,pgfarrows,pgfnodes,pgfautomata,pgfheaps,pgfshade} 
\usepackage{tkz-graph}
\usepackage{pgfplotstable}
\usepackage{pgfplots}
\usepackage{booktabs}

\usepackage{tikz}
\usetikzlibrary{automata, positioning, calc, shapes, arrows, shadows, patterns, plotmarks}
\usepgflibrary{shapes}

\usepackage{tgpagella}   

\newcommand\held[1]{{\fcolorbox{black}{white}{$#1$}}}

\newcommand\heldc[1]{{\underline{#1}}}

\newcommand{\agentSet}{\mathcal{N}}
\newcommand{\resourceSet}{\mathcal{R}}
\newcommand{\tuple}[1]{\langle #1 \rangle}

\newcommand{\CrawlerEx}[2]{\begin{tabular}{c}#1\\$\bullet$ \\#2\end{tabular}}

\newtheorem{definition}{Definition}
\newtheorem{theorem}{Theorem}
\newtheorem{proposition}{Proposition}
\newtheorem{example}{Example}
\newtheorem{observation}{Observation}

\hypersetup{%
	colorlinks=true,
	citecolor=red,
	linkcolor=red,
	anchorcolor=red,
	urlcolor=red,
	pdfstartview=FitH,
	pdfdisplaydoctitle=true,
	pdftoolbar=true,
	pdfmenubar=true
}

\begin{document}

\begin{frontmatter}
\title{An Optimal Procedure to Check Pareto-Optimality in \\ House Markets with Single-Peaked Preferences}

\author[lip6]{Aurélie Beynier}
\author[lip6]{Nicolas Maudet}
\author[illc]{Simon Rey}
\author[lip6]{Parham Shams}

\address[lip6]{LIP6, Sorbonne Université, Paris, France}
\address[illc]{ILLC, University of Amsterdam, Amsterdam, the Netherlands}

\begin{abstract}
	Recently, the problem of allocating one resource per agent with initial endowments (\emph{house markets}) has seen a renewed interest: indeed, while in the domain of strict preferences the Top Trading Cycle algorithm \citep{shapley1974cores} is known to be the only procedure guaranteeing Pareto-optimality, individual rationality, and strategy proofness \citep{ma1994strategy}. However, the situation differs in the single-peaked domain. Indeed, \cite{bade2019matching} presented the \emph{Crawler}, an alternative procedure enjoying the same properties, with the additional advantage of being implementable in obviously dominant strategies. In this paper we further investigate the Crawler and propose  the Diver, a variant which checks optimally whether an allocation is Pareto-optimal for single-peaked preferences, thus improving over known techniques used for checking Pareto-optimality in more general domains. We also prove that the Diver is asymptotically optimal in terms of communication complexity.
\end{abstract}


\end{frontmatter}

\section{Introduction}

\label{sec:intro}

Allocating indivisible resources among a set of agents is a research agenda that has been extensively studied in the recent years. It is a particularly dynamic field in both artificial intelligence (e.g. \cite{brandt2016handbook}) and economics (e.g. \cite{moulin2018fair}). It investigates the issue of fairly and/or efficiently allocating a set of objects to a set of agents while taking into account their preferences. 

In the present paper we focus on the model defined by \cite{shapley1974cores}, called \emph{house market} or \emph{assignment problem}, in which there are exactly as many indivisible resources as agents and where each agent should receive exactly one resource. In house markets, \cite{shapley1974cores} defined the \emph{Top Trading Cycle} (TTC) procedure which has been extensively studied \citep{roth1982incentive}.  
When all preferences expressed as strict linear orders are allowed, the TTC procedure is known to satisfy Pareto-optimality (it is not possible to improve some agents' satisfaction without hurting some others'), strategy-proofness (no one can benefit from reporting non truthful preferences) and individual rationality (no agent receives a house that she likes less than her initial endowment). It is also provably the only procedure enjoying such properties \citep{ma1994strategy} when preferences are strict.

However, preferences frequently exhibit some structures. The domain of \emph{single-peaked} preferences, initially introduced by \cite{black1948rationale} and \cite{arrow1951social} for voting scenarios, is one of the most studied preference domains \citep{moulin1991axioms, elkind2017structured}. It states that there is a common linear order such that all the preferences are decreasing when moving away from the most preferred resource following the order.
This domain is also relevant in resource allocation settings \citep{sprumont1991division, BrunerL15}. For instance, agents may be looking for houses in a street which has a metro station at one of its ends, and a bike rental platform at the other end; the agents' preferences are then likely to be single-peaked depending on the distance to their favourite means of transportation. However, until recently, only a few papers had studied this domain restriction in house markets. \cite{damamme2015power} investigated (distributed) swap dynamics in such settings. 
Motivated by the allocation of time-slots, \cite{hougaard2014assigning} and \cite{aziz2017computational} studied deterministic and probabilistic solutions for the problem of assigning objects to a line, the domain being more restrictive than single-peaked preferences in that case.


Recently, \cite{bade2019matching} presented the \emph{Crawler} procedure for the assignment problem with single-peaked preferences. 
This procedure differs from TTC and also satisfies Pareto-optimality, strategy-proofness and individual rationality. 
As a matter of fact, it is also implementable using \emph{obviously dominant strategies} \citep{li2017obviously} when TTC can not, thus demonstrating the benefit of considering the Crawler in the single-peaked domain restriction.

In this note we define a Crawler-based procedure for checking Pareto-optimality of a given allocation in the single-peaked domain, more efficiently. 
We first provide a brief description of the Crawler (Section \ref{sec:Crawler}), and we analyze its complexity. We then introduce a variant of this procedure, which we call the \emph{Diver} that  can be used to check Pareto-optimality of a given allocation in linear time (Section \ref{sec:CrawlerDiver}).
This improves over known results which resort on cycle detection techniques and thus run in $O(n^2)$ \citep{abraham}. 
The procedure also turns out to be optimal in terms of time and communication complexity.

\section{Preliminaries}
\label{sec:preliminaries}

 We consider a set of agents $\agentSet = \{a_1, \ldots, a_n\}$ and a set of resources $\resourceSet = \{r_1, \ldots, r_n\}$ of the same size. An allocation $\pi = \tuple{\pi_{a_1}, \ldots, \pi_{a_n}}$ is a vector of $\resourceSet^n$ whose components $\pi_{a_i} \in \resourceSet$ represent the single resource allocated to agent $a_i \in \agentSet$. 

Agents are assumed to express their preferences over the resources through complete linear orders. Agent $a_i$'s preferences are denoted by $\succ_{i}$ , where $r_1 \succ_{i} r_2$ means that $r_1$ is strictly preferred over $r_2$. A preference profile $L = \tuple{\succ_i}_{a_i \in \agentSet}$ is then a tuple of all the agents' preferences. 

For a given linear order $\succ$, we use $top(\succ)$ to denote the top-ranked resource. 
Similarly, $snd(\succ)$ refers to the second most preferred resource in $\succ$. With a slight abuse of notation we will write $top(a_i)$ and $snd(a_i)$ to refer to $top(\succ_{i})$ and $snd(\succ_{i})$. When it is not clear from the context we will subscript these notations to specify the resource set considered. For instance $top_R(a_i)$ is the most preferred resource for agent $a_i$ among the resources in $R \subseteq \resourceSet$.

An instance of a resource allocation problem is then a tuple $I = \tuple{\agentSet, \resourceSet, L, \pi^0}$ composed of a set of agents $\agentSet$, a set of resources $\resourceSet$, a preference profile $L$ and an initial allocation $ \pi^0$. 

\medskip

In some settings, natural properties of  the agents' preferences can be identified, thus restricting the set of possible preference orderings. The notion of preference domain formalizes these restrictions. For a set of resources $\resourceSet$, we denote by $\mathcal{L}_\resourceSet$ the set of all linear orders over $\resourceSet$. Any subset $D \subseteq \mathcal{L}_\resourceSet$ is then called a preference domain.

We say that an instance $I = \tuple{\agentSet, \resourceSet, L, \pi^0}$ is defined over a preference domain $D$ if the preferences of the agents belong to $D$.

In this note, we consider the single-peaked domain. In this setting, the agents are assumed to share a common axis $\lhd$ over the resources and with respect to which their preferences are defined.

\begin{definition}
	Let $\resourceSet$ be a set of resources and $\lhd$ a linear order (i.e. the axis) over $\resourceSet$. We say that a linear order $\succ$ is single-peaked with respect to $\lhd$ if we have:
	$$\forall (r_1, r_2) \in \resourceSet^2 \text{ s.t.} \left.\begin{array}{r}
	r_2 \lhd r_1 \lhd top(\succ), \\ 
	or, \enspace top(\succ) \lhd r_1 \lhd r_2
	\end{array}\right\} \Rightarrow r_1 \succ r_2.$$
\end{definition}

In other words, $\succ$ is single-peaked over $\lhd$ if $\succ$ is decreasing on both left and right sides of $top(\succ)$, where left and right are defined by $\lhd$.

For a given linear order $\lhd$, we call $\mathcal{SP}_\lhd$ the set of all the linear orders single-peaked with respect to $\lhd$:
$$\mathcal{SP}_\lhd = \{\succ{} \in \mathcal{L}_\resourceSet \mid{} \succ \text{ is single-peaked w.r.t. } \lhd\}.$$

A preference domain $D$ is called single-peaked if and only if $D \subseteq \mathcal{SP}_\lhd$ for a given $\lhd$. An instance $I$ is said to be single-peaked if it is defined over a single-peaked preference domain.

%

\medskip

Pareto optimality of the outcome  guarantees that no agent can improve her allocation without incurring a loss on at least another agent, while individual rationality guarantees agents have incentive to participate. 

\begin{definition}[Pareto-optimality]
	Let $I = \tuple{\agentSet, \resourceSet, L, \pi^0}$ be an instance. An allocation $\pi$ is said to be Pareto-optimal if there is no other allocation $\pi'$ such that for every agent $a_i \in \agentSet$ either $\pi'_{a_i} \succ \pi_{a_i}$ or $\pi'_{a_i} = \pi_{a_i}$ and there exists at least one agent $a_j \in \agentSet$ such that $\pi'_{a_j} \succ \pi_{a_j}$.
	
	If such allocation $\pi'$ exists, we say that $\pi'$ Pareto-dominates the allocation $\pi$.
\end{definition}

\begin{definition}[Individual rationality]
	For a given instance $I = \tuple{\agentSet, \resourceSet, L, \pi^0}$, an allocation $\pi$ is individually rational if for every agent $a_i \in \agentSet$ we have either $\pi_{a_i} \succ \pi^0_{a_i}$ or $\pi_{a_i} = \pi^0_{a_i}$.
\end{definition}

We illustrate these two concepts on a simple example.

\begin{example}
	\label{ex:firstEx}
	
	Let us consider the following instance with 5 agents and 5  resources. The preferences, presented below, are single-peaked with respect to $r_1 \lhd r_2 \lhd r_3 \lhd r_4 \lhd r_5$. The initial allocation $\pi^0 = \tuple{r_5, r_1, r_3, r_4, r_2}$ is defined by the underlined resources.	
	\begin{align*}
	a_1: \enspace & \held{r_1} \succ_{1} r_2 \succ_{1} r_3 \succ_{1} r_4 \succ_{1} \heldc{r_5} \\
	a_2: \enspace & \held{r_5} \succ_{2} r_4 \succ_{2} r_3 \succ_{2} r_2 \succ_{2} \heldc{r_1} \\
	a_3: \enspace & \held{\heldc{r_3}} \succ_{3} r_2 \succ_{3} r_1 \succ_{3} r_4 \succ_{3} r_5 \\
	a_4: \enspace & \held{\heldc{r_4}} \succ_{4} r_3 \succ_{4} r_2 \succ_{4} r_1 \succ_{4} r_5 \\
	a_5: \enspace & r_4 \succ_{5} r_5 \succ_{5} r_3 \succ_{5} \held{\heldc{r_2}} \succ_{5} r_1 	
	\end{align*}

	The allocation $\pi^0$ is not Pareto-optimal as it is Pareto-dominated by the squared allocation $\held{ \pi} = \tuple{r_1, r_5, r_3, r_4, r_2}$. Note that the allocation $\pi' = \tuple{r_1, r_5, r_3, r_2, r_4}$ would make every agent having at least their third top resource, but would violate individual rationality for agent $a_4$.
\end{example}

\section{The Crawler}
\label{sec:Crawler}

 In house markets under single-peaked preferences, \cite{bade2019matching} recently introduced the Crawler procedure.
The agents are initially ordered along the single-peaked axis according to the resource they initially hold. The first agent is the one holding the resource on the left side of the axis and the last agent is the one holding the resource on the right side of the axis. We denote by $R$ the list of available resources ordered according to the single-peaked axis and $N$ the list of available agents such as the $i^{th}$ agent of the list is the one who  holds the $i^{th}$ resource in $R$.  

The algorithm then screens the agents from left to right\footnote{Note that the algorithm can equivalently be executed from right to left.} and check, for each agent $a_i$, where the peak $top_R(a_i)$ of $a_i$ is:
\begin{enumerate}
	\item  If $top_R(a_i)$ is on her right, the algorithm moves to the next agent on the right. The agent is said to ``pass''. 
	\item  If $a_i$ holds her peak $top_R(a_i)$, then $top_R(a_i)$ is allocated to $a_i$, $a_i$ and $top_R(a_i)$ are then removed from $N$ and $R$. The algorithm restarts screening the agents from the left extremity of the axis. 
	\item If $top_R(a_i)$ is on the left of $a_i$, the agent is allocated her peak $top_R(a_i)$. Let $t^*$ be the index of $top_R(a_i)$ and $t$ the index of $a_i$ (we have $t^* < t$). Then, all the agents between $t^*$ and $t - 1$ receive the resource held by the agent on their right (the resources ``crawl'' towards left).  $a_i$ and $top_R(a_i)$ are then removed from $N$ and $R$.   The algorithm restarts screening the agents from the left extremity of the axis. 
\end{enumerate}
The algorithm terminates once $N$, and thus $R$, are empty. 

\medskip

A formal description of the procedure is given in Algorithm \ref{algo:Crawler}. Note that we make use of the sub-procedure $pick(a_{t^*},r,N,R,\pi)$ which simply assigns the resource $r$ to the given agent $a_{t^*}$  in the allocation $\pi$, and then removes the agent and the resource from the lists of available agents and resources, $N$ and $R$ respectively. Since the list of resources is ordered following the single-peaked axis and the $i^{th}$ agent in $N$ corresponds to the owner of the  $i^{th}$ resource in $R$, the removal of $r$ and $a_{t^*}$  is in fact equivalent to assigning $r$ to agent $a_{t^*}$ and crawling the resources from right to left.

\begin{algorithm}[t]
	\DontPrintSemicolon
	\KwIn{An instance $I = \tuple{\agentSet, \resourceSet, L, \pi^0}$ single-peaked with respect to $\lhd$}
	\KwOut{An allocation $\pi$}
	$\pi \gets$ empty allocation \;
	$R \gets \resourceSet$: list of resources sorted accordingly to $\lhd$\;
	$N \gets \agentSet$: list of agents such that the $i^{th}$ agent is the one who initially holds the $i^{th}$ resource in $R$\;
	\While {$N \neq \emptyset$} {
		$t^* \gets |N|$\;
		\For {$t = 0$ to $|N| - 1$} {
			\If(\tcc*[h]{no crawl}) {$r_t \succ_{t} r_{t + 1}$} {
				$t^* \gets t$\;
				Break\;
			}
		}
		$r \gets top_R(a_{t^*})$ \;
		$pick(a_{t^*},r,N,R, \pi)$\;
	}
	\Return $\pi$
	\caption{The Crawler procedure}
	\label{algo:Crawler}
\end{algorithm}

Let us illustrate the execution of the Crawler on the instance of Example \ref{ex:firstEx}.

\begin{example}
	Let us return to Example \ref{ex:firstEx}. The execution of the Crawler is depicted in Figure \ref{fig:exCrawler}. The initial allocation is presented in the box 1. Initially, agent $a_3$ is the first agent whose top is not on her right, she thus receives her top $r_3$ (box 2). The second step matches agent $a_4$ to $r_4$ (box 3). On the third step, agents $a_2$ and $a_5$ both have their top on the right but the last agent $a_1$  has her top on her left, she is then matched to her top $r_1$ (box 4). Agent $a_5$ is matched to $r_5$ on the fourth step (box 5). Finally $a_2$ is assigned resource $r_2$ (box 6).
	\begin{figure}
		\centering
		\resizebox{\linewidth}{!}{
			\begin{tikzpicture}[shorten <= 5pt, shorten >= 5pt]
			
			\node[draw, dashed, label = 180:\boxed{1}] (step0) {\begin{tikzpicture}[node distance = 0.2cm]
				\node[] (step1r1) {\CrawlerEx{$r_1$}{$a_2$}};
				\node[right = of step1r1] (step1r2) {\CrawlerEx{$r_2$}{$a_5$}};
				\node[right = of step1r2] (step1r3) {\CrawlerEx{$r_3$}{$a_3$}};
				\node[right = of step1r3] (step1r4) {\CrawlerEx{$r_4$}{$a_4$}};
				\node[right = of step1r4] (step1r5) {\CrawlerEx{$r_5$}{$a_1$}};
				\end{tikzpicture}};

			\node[draw, dashed, right = of step0, label = 180:\boxed{2}] (step1) {\begin{tikzpicture}[node distance = 0.2cm]
				\node[] (step1r1) {\CrawlerEx{$r_1$}{$a_2$}};
				\node[right = of step1r1] (step1r2) {\CrawlerEx{$r_2$}{$a_5$}};
				\node[right = of step1r2] (step1r3) {\held{\CrawlerEx{$r_3$}{$a_3$}}};
				\node[right = of step1r3] (step1r4) {\CrawlerEx{$r_4$}{$a_4$}};
				\node[right = of step1r4] (step1r5) {\CrawlerEx{$r_5$}{$a_1$}};
				\end{tikzpicture}};
			
			\node[draw, dashed,  below = of step0, label = 180:\boxed{3}] (step2) {\begin{tikzpicture}[node distance = 0.1cm]
				\node[] (step2r1) {\CrawlerEx{$r_1$}{$a_2$}};
				\node[right = of step2r1] (step2r2) {\CrawlerEx{$r_2$}{$a_5$}};
				\node[right = of step2r2] (step2r3) {\held{\CrawlerEx{$r_3$}{$a_3$}}};
				\node[right = of step2r3] (step2r4) {\held{\CrawlerEx{$r_4$}{$a_4$}}};
				\node[right = of step2r4] (step2r5) {\CrawlerEx{$r_5$}{$a_1$}};
				\end{tikzpicture}};
			
			\node[draw, dashed, below = of step2, label = 180:\boxed{5}] (step4) {\begin{tikzpicture}[node distance = 0.1cm]
				\node[] (step4r1) {\held{\CrawlerEx{$r_1$}{$a_1$}}};
				\node[right = of step4r1] (step4r2) {\CrawlerEx{$r_2$}{$a_2$}};
				\node[right = of step4r2] (step4r3) {\held{\CrawlerEx{$r_3$}{$a_3$}}};
				\node[right = of step4r3] (step4r4) {\held{\CrawlerEx{$r_4$}{$a_4$}}};
				\node[right = of step4r4] (step4r5) {\held{\CrawlerEx{$r_5$}{$a_5$}}};
				\end{tikzpicture}};
			
			\node[draw, dashed, right = of step4, label = 180:\boxed{6}] (step5) {\begin{tikzpicture}[node distance = 0cm]
				\node[] (step5r1) {\held{\CrawlerEx{$r_1$}{$a_1$}}};
				\node[right = of step5r1] (step5r2) {\held{\CrawlerEx{$r_2$}{$a_2$}}};
				\node[right = of step5r2] (step5r3) {\held{\CrawlerEx{$r_3$}{$a_3$}}};
				\node[right = of step5r3] (step5r4) {\held{\CrawlerEx{$r_4$}{$a_4$}}};
				\node[right = of step5r4] (step5r5) {\held{\CrawlerEx{$r_5$}{$a_5$}}};
				\end{tikzpicture}};
			
			\node[draw, dashed, above = of step5, label = 180:\boxed{4}] (step3) {\begin{tikzpicture}[node distance = 0cm]
				\node[] (step3r1) {\held{\CrawlerEx{$r_1$}{$a_1$}}};
				\node[right = of step3r1] (step3r2) {\CrawlerEx{$r_2$}{$a_2$}};
				\node[right = of step3r2] (step3r3) {\held{\CrawlerEx{$r_3$}{$a_3$}}};
				\node[right = of step3r3] (step3r4) {\held{\CrawlerEx{$r_4$}{$a_4$}}};
				\node[right = of step3r4] (step3r5) {\CrawlerEx{$r_5$}{$a_5$}};
				\end{tikzpicture}};
			
	\end{tikzpicture}}
	\caption{The Crawler procedure run on Example \ref{ex:firstEx}. Each dashed box corresponds to a step and a pair resource agent is boxed if the resource has been allocated to the agent.}
	\label{fig:exCrawler}
\end{figure}
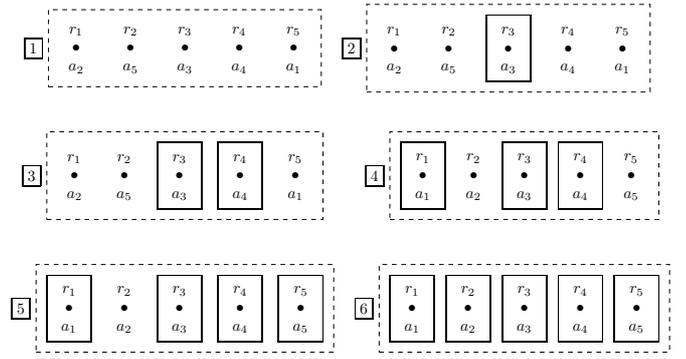
\end{example}

\medskip
As observed by \cite{bade2019matching}, the Crawler always terminates. It is easy to see that it runs in quadratic time. 

\begin{proposition}
	The Crawler procedure terminates and its complexity is in $\mathcal{O}(n^2)$ where $n$ is the number of agents and objects.
\end{proposition}

\begin{proof}
	Termination is proved by observing that $|N|$ is strictly decreasing at each step of the main while loop. This loop is applied at most $n$ times and each step of the loop requires at most $\mathcal{O}(n)$ elementary operations. The time complexity is then in $\mathcal{O}(n^2)$.
\end{proof}

We conclude this section by studying the communication requirement of the procedure \citep{kushilevitz1996communication}. We are here interested in the amount of information communicated from the agents to the center. 

\begin{proposition}
	The crawler requires at most $n(n+1)/2 + n \log n$ bits of communication. 
\end{proposition}

\begin{proof}
	The crawler runs in $n$ rounds. At each round $i$, every remaining agent is being asked whether she wishes to pass (answered using 1 bit) or to designate a resource on the left she wants to get (answered in at most $\log n$ bits). Overall, for the $i$-th round, at most $n - i$ agents will pass, and clearly only one agent designates a resource. Thus the protocol requires in the worst case $\sum_{i=1}^n [(n-i) + \log n] = n(n+1)/2 + n \log n$, which is in $\mathcal{O}(n^2)$. 
\end{proof}

Since communicating the full preference lists requires $\mathcal{O}(n^2 \log n)$, this protocol does save some communication, even asymptotically, compared to the naive protocol.

\section{Optimally checking Pareto-optimaility: the Diver}
\label{sec:CrawlerDiver}

We now turn to the question of whether we can gain advantage from the single-peaked domain in order to check Pareto-optimality more efficiently. Recall that in the domain allowing any linear order, this can be done in quadratic time \citep{abraham}.

More formally, given an instance $I = \tuple{\agentSet, \resourceSet, L, \pi^0}$ single-peaked with respect to $\lhd$, we consider the problem \textsc{CheckPO} whose answer is yes if and only if $\pi^0$ is Pareto-optimal. We assume that $\lhd$ is known by the agents. Without loss of generality, we consider that $\forall a_i \in \agentSet, \pi^0_{a_i} = r_i$.

First, observe that the Crawler indeed returns the initial allocation when it is Pareto-optimal. 
\begin{observation}
	Let $\pi$ be an allocation, $\pi$ is Pareto-optimal if and only if the Crawler returns $\pi$ when applied to $\pi$ as the initial allocation.
\end{observation}

\begin{proof}
	Any allocation $\pi$ returned by the Crawler is Pareto-optimal \citep{bade2019matching}. Thus, if $\pi$ is not Pareto-optimal a different allocation is returned.  
	On the contrary, if $\pi$ is Pareto-optimal, and since the Crawler is individually rational, no trading cycle will be performed during the execution. The Crawler thus returns $\pi$.
\end{proof}

However, this procedure does not enjoy better complexity guarantees than the ones not specific to single-peaked domains, as its worst-case time complexity is in $\mathcal{O}(n^2)$. A worst-case instance can be described as follows: suppose that all the agents (ordered from left to right), have the next resource on their right as their top, except for the last one who likes her own resource. In that case, at each step, the Crawler would go through all the agents before realizing that the last one wants to keep her resource.

\medskip

We propose a variant of the Crawler, called the \emph{Diver}, which allows to check Pareto-optimality of the initial allocation more efficiently.
The key difference with the Crawler is that the \emph{Diver} procedure does not start a new screening once an agent picks a resource: it only checks whether the last agent who was happy to crawl for this resource now agrees to \emph{dive} to the next one. The  \emph{Diver} thus proceeds in a \emph{single screening} of the agents. 

At each step, the central entity asks the agent whether she wishes to:  
\begin{enumerate}
	\item[(1)] pick her current resource; 
	\item[(2)] pass (expressing that she is happy to dive to the next resource); or 
	\item[(3)] pick a smaller resource. 
\end{enumerate}
Note that, each time an agent picks a resource, the central entity communicates this information to the other remaining agents so that they can update their list of available resources. 

In case (1), the agent (and her resource) are removed and we enter a sub-protocol called \emph{backtrack-call} in which the previous agents are asked one by one whether they still agree to dive to the next resource. This sub-protocol stops as soon as one agent says yes, or when there are no more agents left to consider. All the agents who said 'no' pick their current resources and are themselves removed together with their resource. 

Whenever an agent is happy to dive to the next resource, the diver simply moves on to the next agent. This is case (2).

As soon as an agent says she wants a smaller resource, the protocol stops and returns 'not PO'. This corresponds to case (3). Note that in this case, we have the guarantee that there is indeed a better resource available, otherwise the agent would have picked her own resource. 

If the screening goes through all the agents, then all the agents left the protocol with their own resource, and the protocol returns 'PO'. 

\medskip

The protocol is formally described in Algorithm \ref{algo:diver}. 
The sub-procedure $pick(a_i,r)$ simply assigns resource $r$ to agent $a_i$, while $pick(a_i,r,D)$ does the same, and removes agent $a_i$ from the list $D$ of agents who crawl or dive.  

\begin{algorithm}[t]
	\DontPrintSemicolon
	\KwIn{An instance $I = \tuple{\agentSet, \resourceSet, L, \pi^0}$ single-peaked with respect to $\lhd$}
	\KwOut{PO if $\pi^0$ is Pareto-optimal and not PO otherwise}
	$\pi \gets$ list of pairs $(a_i,r_i)$ such that agent $a_i$ holds resource $r_i$ in $\pi^0$, sorted according to $\lhd$ for the resources  \;
	$D \gets \emptyset$: list of agents who crawl or dive\;
	
	\For { $(a_i,r_i)$ in $\pi$} {
		\uIf(\tcc*[h]{pick your top}) {$top_R(a_{i}) = r_{i}$} {	
			$pick(a_{i}, r_{i})$ \;
			\For {$a_j$ in $reverse(D)$} {
				\uIf(\tcc*[h]{if you don't dive, pick your resource}){$r_{j} \succ_{a_{j}} r_{i + 1}$}
				{$pick(a_{j},r_{j},D)$\;}
				\Else{Break \;}
				
		}}
		
		\uElseIf(\tcc*[h]{your top is on your left: not PO}){$r_{i} \succ_{a_{i}} r_{i + 1}$}{
			\Return not PO \;}
		\Else(\tcc*[h]{crawl}){$D \gets D.append(a_i)$}

	}
	\Return PO \;
	\caption{The Diver procedure}
	\label{algo:diver}
\end{algorithm}

\begin{example}
	Coming back to Example \ref{ex:firstEx}, by applying the Diver to the initial allocation $\pi^0$,  the agents are first sorted as follows:
	
	\begin{center}	
		\resizebox{0.7\linewidth}{!}
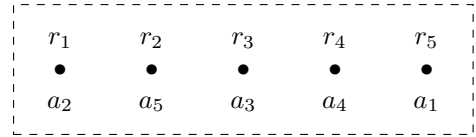
{
			\begin{tikzpicture}[shorten <= 5pt, shorten >= 5pt]
			
			\node[draw, dashed] (step0) {\begin{tikzpicture}[node distance = 0.2cm]
				\node[] (step1r1) {\CrawlerEx{$r_1$}{$a_2$}};
				\node[right = of step1r1] (step1r2) {\CrawlerEx{$r_2$}{$a_5$}};
				\node[right = of step1r2] (step1r3) {\CrawlerEx{$r_3$}{$a_3$}};
				\node[right = of step1r3] (step1r4) {\CrawlerEx{$r_4$}{$a_4$}};
				\node[right = of step1r4] (step1r5) {\CrawlerEx{$r_5$}{$a_1$}};
				\end{tikzpicture}};	
			
	\end{tikzpicture}}	
\end{center}

The Diver screens the agent from left to right and asks each agent her wish:
\begin{enumerate}
	\item $a_2$ passes;
	\item $a_5$ passes;
	\item $a_3$ picks her current resource,  $a_5$ still agrees to pass; 
	\item $a_4$   picks her current resource,  $a_5$ still agrees to pass; 
	\item $a_1$ wants to pick a smaller resource ($r_1$) $\rightarrow$ the Diver returns 'not PO'.
\end{enumerate}
Indeed, this allocation is dominated by $\langle r_1, r_2, r_3, r_4, r_5 \rangle$. 

Now let us consider the allocation $\pi = \langle r_1, r_5, r_2, r_4, r_3 \rangle$ leading to the following order:

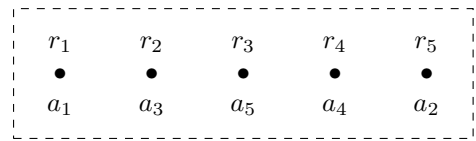
\begin{center}		
	\resizebox{0.7\linewidth}{!}{
		\begin{tikzpicture}[shorten <= 5pt, shorten >= 5pt]
		
		\node[draw, dashed] (step0) {\begin{tikzpicture}[node distance = 0.2cm]
			\node[] (step1r1) {\CrawlerEx{$r_1$}{$a_1$}};
			\node[right = of step1r1] (step1r2) {\CrawlerEx{$r_2$}{$a_3$}};
			\node[right = of step1r2] (step1r3) {\CrawlerEx{$r_3$}{$a_5$}};
			\node[right = of step1r3] (step1r4) {\CrawlerEx{$r_4$}{$a_4$}};
			\node[right = of step1r4] (step1r5) {\CrawlerEx{$r_5$}{$a_2$}};
			\end{tikzpicture}};		
\end{tikzpicture}}
\end{center}

Again, the Diver screens the agent from left to right and asks each agent her wish:
\begin{enumerate}
\item $a_1$ picks her current resource;
\item $a_3$ passes;
\item $a_5$ passes; 
\item $a_4$  picks her current resource,  $a_5$ still agrees to pass; 
\item $a_2$ picks her current resource,  $a_5$ picks her current resource, $a_3$ picks her current resource. All the agents have left with their resource and the Diver returns 'PO'.
\end{enumerate}

\end{example}

Next, we prove the correctness of the Diver and show that it runs in $\mathcal{O}(n)$. 

\begin{theorem}
The Diver always terminates, runs in $\mathcal{O}(n)$ and returns whether the initial assignment is Pareto-optimal or not. 
\end{theorem}

\begin{proof}
Termination is obvious since the procedure proceeds in a single main screening of the finite set of resources. 
We first show that the procedure is sound. 
First observe that when the Diver returns 'PO', all the agents must have picked their initial resource. Indeed, consider the last agent in the order: this agent picked her resource (otherwise the procedure would have returned 'not PO'). But now the agent on her left must also have picked her resource (as there is no more possibility to dive), and so on until there are no agents remaining. 
Now, following the argument used in \cite{bade2019matching}, 
consider all the agents who picked their resource during this process, in the order they picked it: they clearly have all picked their best available resource. 
The obtained matching is thus indeed Pareto-optimal. 
On the other hand, when the Diver returns 'not PO', there is indeed an improving cycle, consisting of the agent (say, $a_j$) who chose a resource on her left, and all the agents, from the owner of this resource to $a_j$, who are not matched yet. 

\medskip

In terms of complexity, sorting the agents according to the single-peaked order can be done in $\mathcal{O}(n)$ using counting sort \citep[Section 8.2]{cormen2009introduction}. 
Now for the main loop of the procedure: in the reverse loop, note that if $k+1$ agents are screened backwards, then $k$ agents are removed for good. Thus through the entire procedure the reverse loop involves $\mathcal{O}(n)$ steps, and thus the main loop takes $\mathcal{O}(n)$ as well, which gives us the linear time complexity.   
\end{proof}




The same line of analysis allows us to derive a result regarding the amount of communication induced by the procedure, as can be done for other social choice mechanisms, see \emph{e.g.} \cite[Chapter 10]{brandt2016handbook}. 
The Diver only requires a linear (in the number of agents) number of bits to be communicated from the agents.

\begin{proposition}
The Diver requires 4n bits of communication. 
\end{proposition}

\begin{proof}
The key is to observe that sub-protocol backtrack-call requires overall $n + n$ bits, as there may only be $n$ agents saying 'no' and $n$ agents saying 'yes' throughout the whole run of the Diver.
In the main loop of the protocol, the query requires 2 bits to be answered (as there are three possible answers). This makes overall $2n + 2n = 4n$ bits, thus, only $\mathcal{O}(n)$ bits.
\end{proof}


We now show that the Diver is asymptotically optimal, both in terms of time and communication complexity. 
In fact, a simple adversarial argument suffices to show that any algorithm solving this problem must query at least $n-1$ agents and thus read an input of this size at least. 

\begin{proposition}
	In the single-peaked domain, the time complexity of \textsc{CheckPO} is $\Omega(n)$.
\end{proposition}

\begin{proof}
	Simply take as adversarial input an instance where every agent receives initially her preferred house. No procedure can answer \textsc{CheckPO} before querying less than $n-1$ agents. Indeed, as long as two agents remain to be queried, it could be that they form a trading cycle. 
\end{proof}



As each query requires at least one bit to be answered, this immediately implies a similar bound on the communication complexity. This can alternatively be shown by exhibiting a straightforward \emph{fooling set} \citep{kushilevitz1996communication}, which we provide for completeness. We consider strict preferences for the agents written $\succ_i^L$ in a profile $L$, and by a slight abuse of notation we write  $\succ_i^{L|L'}$ to say that $\succ_i$ is the preference of agent $a_i$ in either $L$ or $L'$.  In our context, the fooling set will be a collection of profiles $\mathcal{F} = \langle L_1, \dots L_K \rangle$ such that:
\begin{enumerate}
\item for any $i \in \{1,\dots K\}$ \textsc{CheckPO}'s answer on $\tuple{\agentSet, \resourceSet, L_i, \pi^0}$ is yes.
\item for any $i \not =j$, there exists $L' = \langle \succ_1^{L_i|L_j}, \dots \succ_n^{L_i|L_j} \rangle$, such that \textsc{CheckPO}'s answer on $\tuple{\agentSet, \resourceSet, L', \pi^0}$ is no.  
\end{enumerate}
By a standard result in communication complexity, it is known that $\log |\mathcal{F}|$ is a lower bound on the communication complexity of the problem \citep{kushilevitz1996communication}.

\begin{proposition}
	\label{prop:LowerBoundComComplex}
	In the single-peaked domain, the communication complexity of \sc{CheckPO} is $\Omega(n)$. 
\end{proposition} 

\begin{proof}
Let us call a \emph{consensual profile} the profile where $\succ_i = \succ_j$ for any $(i,j)$, i.e. all the agents have the same linear orders over the resources. The consensual linear order will be denoted by $\succ$. We claim that the set $\mathcal{F}$ of the $2^{n-1}$ (single-peaked) consensual profiles constitutes a fooling set. 

To show this, first observe that in any such profile, the original assignment $\pi^0$ is Pareto-optimal. Indeed, in a consensual profile, no trading cycle is possible. Hence the aforementioned condition 1. of a fooling set is satisfied.  

Now to show that we can fool the function, consider any pair of profiles $(L_i,L_j)$. As these profiles are different, there must exist at least one pair of resources $(r_p,r_q)$ such that $r_p \succ r_q$ in $L_i$, while $r_q \succ r_p$ in $L_j$ (it is true for all agents since the profiles are consensual). Now consider the agent $a_p$ (resp. $a_q$) holding $r_p$ (resp. $r_q$) in $\pi^0$ and  a mixed profile $L'$ such that:
$$\forall k \in \agentSet, \succ_k^{L'} = \left\{\begin{array}{ll}
\succ_k^{L_i} & \text{if } k \neq p, \\
\succ_k^{L_j} & \text{if } k = p.
\end{array}\right.$$
Hence,  $a_p$ and $a_q$ have opposite preferences for $r_p$ and $r_q$ and would prefer to swap, i.e. \textsc{CheckPO}'s answer on $\tuple{\agentSet, \resourceSet, L', \pi^0}$ is no. This concludes the proof.
\end{proof}




The Diver is thus asymptotically optimal in terms of time and communication complexity.

\section{Future work}
\label{sec:conclusion}

A natural extension of the domain studied is to allow for indifferences in preferences. 
In the universal domain where indifferences can be expressed, \cite{aziz2012housing} defined a set of Pareto-optimal procedures generalizing TTC in that setting, which however include procedures that are not strategy-proof. \cite{plaxton2013simple} and \cite{saban2013house} independently proposed general frameworks for efficient and strategy-proof generalization of the TTC procedure with indifferences.
In her paper, \cite{bade2019matching} defines the ``circle crawling'' procedure, which enjoys the same properties as the crawler. 
It would be interesting to study whether a variant in the spirit of the Diver could be adapted in that setting as well. 



\bibliographystyle{elsarticle-harv}


\end{document}